\documentclass[conference]{IEEEtran}

\usepackage{amsmath}
\usepackage{amsthm}
\newtheorem{thm}{Theorem}

\newtheorem{pro}[thm]{Proposition}
\newcommand{\Tr}{Tr}

\usepackage{amssymb}

\usepackage{bm}
\newcommand{\vect}[1]{\boldsymbol{\mathbf{#1}}}

\DeclareSymbolFont{bbold}{U}{bbold}{m}{n}
\DeclareSymbolFontAlphabet{\mathbbold}{bbold}

\usepackage{accents}
\newcommand{\ubar}[1]{\underaccent{\bar}{#1}}

\allowdisplaybreaks

\usepackage{hyperref}

\usepackage{booktabs}

\usepackage{mathtools}

\usepackage{tikz}

\usepackage{subcaption}

\usepackage{flushend}

\ifCLASSINFOpdf
\else
\fi
\usepackage{url}

\usepackage[]{fancyhdr} %
\newcommand{\changefont}{\fontsize{9}{9}\selectfont}
\IEEEoverridecommandlockouts
\begin{document}

%
\title{Enabling Grid-Aware Market Participation of Aggregate Flexible Resources}

\author{\IEEEauthorblockA{Bai Cui, Ahmed Zamzam, Andrey Bernstein}
\IEEEauthorblockA{National Renewable Energy Laboratory\\
Golden, CO, USA\\ \texttt{\{bai.cui, ahmed.zamzam, andrey.bernstein\}@nrel.gov}}
\thanks{This work was authored by the National Renewable Energy Laboratory, operated by Alliance for Sustainable Energy, LLC, for the U.S. Department of Energy (DOE). The views expressed in the article do not necessarily represent the views of the DOE or the U.S. Government. The U.S. Government retains and the publisher, by accepting the article for publication, acknowledges that the U.S. Government retains a nonexclusive, paid-up, irrevocable, worldwide license to publish or reproduce the published form of this work, or allow others to do so, for U.S. Government purposes.}}


%





\maketitle
\thispagestyle{fancy}
\pagestyle{fancy}


\begin{abstract}
Increasing integration of distributed energy resources (DERs) within distribution feeders provides unprecedented flexibility at the distribution-transmission interconnection. With the new FERC 2222 order, DER aggregations are allowed to participate in energy market. To enable market participation, these virtual power plants need to provide their generation cost curves. 
This paper proposes efficient optimization formulations and solution approaches for the characterization of hourly as well as multi-time-step generation cost curves for a distribution system with high penetration of DERs. Network and DER constraints are taken into account when deriving these cost curves, and they enable active distribution systems to bid into the electricity market. The problems of deriving linear and quadratic cost curves are formulated as robust optimization problems and tractable reformulation/solution algorithm are developed to facilitate efficient calculations. The proposed formulations and solution algorithm are validated on a realistic test feeder with high penetration of flexible resources.
\end{abstract}



%
\IEEEpeerreviewmaketitle

\section{Introduction}

The increased share of solar, wind, energy storage, and dispatchable load in the generation mix is overhauling the traditional paradigm of energy systems operations. While scenarios such as renewables meeting total demand seemed far-fetched until recently, system operators in Southern Australia and the State of California have experienced situations where renewable met at least 80\% of total demand in 2020 and 2019, respectively \cite{aemo_2021,californiaISO}. In addition, with targets of 50\% emissions reductions by 2030, and reaching net zero emissions economy-wide by no later than 2050\cite{white_house_2021}, increasing integration of renewables and flexible resources is expect to rise significantly across the globe in the near future.

Increasing installation of distributed energe sources within distribution systems is transforming the traditionally passive networks into virtual power plants (VPP). That is, the aggregate power consumption/generation of distribution networks can be shaped with flexibility provided from renewables, energy storage systems, and flexible loads, to behave as a power plant injecting power into the transmission system. With FERC Order No. 2222 \cite{ferc_2020}, VPPs or aggregators are allowed to participate in energy markets. In order to participate in day-ahead energy market, generators need to provide their (piecewise linear) daily generation cost curves in one-hour resolution, which describe the relationship between power output and the corresponding cost at each hour of the day. For conventional generators, deriving the cost curve is relatively simple, as the cost is a direct reflection of the generator's own characteristics and fuel costs. On the other hand, VPPs participation necessitates the design of corresponding cost curves that maximizes the operational objectives of aggregators while being cognizant of grid operational constraints.

Generally speaking, designing the VPP cost curve is challenging due to 1) the power output of a VPP can be collectively realized by the downstream distributed energy resources (DERs) in different ways, making the mapping between the power output and actual generation cost ill-defined, and 2) the capability of energy storage units to save their energy generation for later use complicates the mapping between the generation cost and the power output of the VPP. That is, the cost curve is time-coupled unlike conventional generators where the power generated by some DERs at one time may be stored and only supplied to the grid at a later time. The time coupling in this case is much more complicated than simple constraints such as ramp constraint which can be provided to market operators to include in their market clearing optimizations. Consequently, the goal of the paper is to overcome these issues and develop a computationally efficient approach to generate \emph{feasible} daily generation cost curves of a VPP. 

Different approaches have been utilized in the literature to exploit the flexibility within distribution systems. A commonly used method is demand response which has been proposed extensively in the literature where nodes with flexible resources can monetize this flexibility \cite{deng2015survey, jordehi2019optimisation}. A widespread demand response paradigm focused on providing variable rates for energy to solicit a specific load response from customers, which includes time-of-use pricing \cite{datchanamoorthy2011optimal} and variable real-time pricing \cite{allcott2009real} schemes. Another paradigm is incentive-based demand response where the customers are incentivized to respond to load reduction instructions from system operators and sometimes penalized in case of not responding \cite{chai2019incentive, aalami2010demand}. Some of these demand response strategies have been deployed in practice due their ability to lower demand especially during peak hours. However, the heterogeneity of flexibility forms, especially when resources are aggregated, opens new more organized venues for monetizing the flexibility. One of these venues is for aggregators to participate in energy markets which necessitates characterizing the incurred cost of providing any specific aggregate power generation profile.

A novel approach to price flexibility in demand side was proposed in \cite{werner2021pricing} where shiftable demands were directly integrated into traditional energy market optimization. The shiftable demands are loads that need to consume a predefined energy amount regardless of how this amount is divided between time intervals. While providing overall system efficiency benefits, such demand is not common in practice which limits the applicability of the approach. The approach does not consider scenarios of aggregated demand where there might be multiple disaggregation strategies to an aggregate load profile which have different operational costs. In such scenarios, flexibility should be priced to ensure profitability of the aggregator operations. Another aspect of the pricing is to ensure that realizing a specific load profile does not lead to network operational constraints violations. This issue can be resolved by resorting to flexibility characterizations that encapsulate some network operational constraints such the high-dimensional box characterization in \cite{chen2019aggregate}, the ellipsoidal characterization in \cite{cui2021network}, or other characterizations \cite{lopez2021quickflex,chen2021leveraging}.

The main goal of this paper is to identify flexibility cost curves that define the price of any aggregate power generation profile from the VPP. Note that any aggregate power profile is often realizable using infinitely many disaggregated power consumption/generation of the aggregated flexible resources. Accordingly, the cost of any aggregated power profile can be different especially under asymmetrical energy prices for individual resources, or asymmetrical prices for buying and selling energy. For this reason, in this paper, the aggregate flexibility price is designed to guarantee profitability of any specific aggregate power profile under any possible disaggregation. The paper formulates this task as a multi-stage robust optimization problem that determines parameters of the lowest cost curve guaranteeing profitability. In particular, the paper utilizes a linear and a quadratic cost curves that characterize the price of an aggregate power profile --- not separate cost curves for each time interval. Although, as will be discussed in the paper, compatible single-time-step (e.g., hourly) cost curves can be derived. Using duality techniques, the optimization problem is translated into a single-stage optimization problem that can be solved efficiently to identify the cost curve parameters. To the best of the authors' knowledge, this represents the first effort to characterize the operational cost of realizing an aggregate power profile from aggregated heterogeneous flexible resources.

The remainder of this paper is organized as follows. In Section \ref{sect:model}, the distribution feeder and aggregation models are presented. Then, the proposed optimization formulation is developed in Section \ref{sect:formulation}. The proposed approach is verified through multiple simulations in Section \ref{sect:simulation}, and the paper is concluded in Section \ref{sect:conclusion}.

\section{System Modeling} \label{sect:model}

Consider a multiphase radial distribution network where the set $\mathcal{N}$ collects the nodes, i.e., $\mathcal{N} = \{0, 1, \ldots, N\}$. We denote the substation bus as bus-$0$. Then, we define the set $\mathcal{N}^+ := \mathcal{N}\backslash\{0\}$. To simplify the notations, it is assumed that all the buses $l\in \mathcal{N}$ feature three-phase connections, namely, $a$, $b$, and $c$. The set ${\Phi}=\{a, b, c\}$ and ${\Phi}_{\Delta} = \{ab, bc, ca\}$ are defined to include the different phases. In addition, the voltage magnitudes at bus $k$ is denoted by $\vect{v}_k \in \mathbb{R}^3$. The vector $\vect{v}\in\mathbb{R}^{3N}$ collects the voltage magnitudes all phases in the network except for the slack bus.
Also, the time horizon is discretized into $T$ time periods, where each period has an equal length of $\tau$.
Three types of DERs are considered in this paper which have diverse operational characteristics, namely, photovoltaic (PV) inverters, battery systems, and HVAC systems. The proposed aggregation framework can also be generalized to account for DER models.

\subsection{DER Models}

\subsubsection{PV Inverters}

The set of phases equipped with PV inverters are collected in the a set denoted by $\mathcal{R}$. The active and reactive power injections at time $t$ for a PV unit installed at phase $\phi$ at bus $k$ are constrained as follows:
\begin{align} \label{eq:ren-const}
& 0 \leq p_{k, \phi}^{(R)}(t) \leq \bar{P}_{k, \phi}^{(R)}(t), && \forall t \\
& q_{k, \phi}^{(R)}(t) = \psi_k^{(R)} p_{k, \phi}^{(R)}(t) && \forall t
\end{align}
where $\bar{P}_{k, \phi}^{(R)}(t)$ denotes the available active power at this PV unit at time $t$, and the fixed power factor of the PV unit is denoted by $\psi_k^{(R)}$. Note that PV united with variable power factor can be considered by considering $\psi_k^{(R)}$ as a variable that can constrained which does not introduce undesirable nonlinearity in the model.

\subsubsection{Energy Storage Units}

The set of phases connected to energy storage systems are collected in the set $\mathcal{B}$. The power output of energy storage units installed at bus $k$ and phase $\phi$ is constrained as follows:
\begin{align}\label{eq:bat-rate-const}
    & \ubar{P}_{k, \phi}^{(B)} \leq p_{k, \phi}^{(B)}(t) \leq \bar{P}_{k, \phi}^{(B)}, &&\forall t\\
    & q_{k, \phi}^{(B)}(t) = \psi_k^{(B)} p_{k, \phi}^{(B)}(t), && \forall t
\end{align}
where $\ubar{P}_{k, \phi}^{(B)}$ and $\bar{P}_{k, \phi}^{(B)}$ denote the maximum charging and discharging rates, and $\psi_k^{(B)}$ denotes the fixed power factor of the energy storage unit. Additionally, it is assumed that the state of charge of the energy storage unit  satisfies the limits:
\begin{align}\label{eq:bat-SoC}
    & e_{k, \phi} (t) = \kappa_k e_{k, \phi}(t-1) - \tau p_{k, \phi}^{(B)}(t), && \forall t \\
    & \ubar{e}_{k, \phi} \leq e_{k, \phi}(t) \leq E_{k, \phi}, && \forall t
\end{align}
where $e_{k, \phi} (t)$ denotes the battery state of charge; $0 < \kappa_k \le 1$ is the storage efficiency factor. Also, $\ubar{e}_{k, \phi}$ and $E_{k, \phi}$ represent the minimum allowed state of charge and the capacity of the energy storage unit, respectively.

\subsubsection{HVAC Systems}
The phases equipped with HVAC systems are collected in the set $\mathcal{H}$. Similar to energy storage units, the power consumption of the system install at bus $k$ and phase $\phi$ at time $t$ abides by the following constraints:
\begin{align} \label{eq:hvac-const}
& 0 \leq p_{k, \phi}^{(h)}(t) \leq \bar{P}_{k, \phi}^{(h)}(t), \quad q_{k, \phi}^{(h)}(t) = \psi_k^{(h)} p_{k, \phi}^{(h)}(t) && \hspace{-0.2in} \forall t \\
& \ubar{H}_k \le H_k^{\mathrm{in}}(t) \le \bar{H}_k, && \hspace{-0.2in} \forall t \\
& H_k^{\mathrm{in}}(t) = H_k^{\mathrm{in}}(t-1) + \alpha_k (H_k^{\mathrm{out}}(t) - H_k^{\mathrm{in}}(t-1)) \nonumber\\
& \hspace{1.5in} + \tau \beta_k p_k(t)^{(h)}(t), && \hspace{-0.2in} \forall t \label{eq:hvac-const:c}
\end{align}
where $\bar{P}_{k, \phi}^{(h)}(t)$ denotes the active power capacity of the HVAC system at time $t$, $\psi_k^{(h)}$ denotes the fixed power factor, and $H^{\mathrm{in}}(t)$ and $H^{\mathrm{out}}(t)$ are the indoor and outdoor temperatures at time $t$. Equation \eqref{eq:hvac-const:c} represents the indoor temperature dynamics, where $\alpha_k$ and $\beta_k$ are the parameters specifying the thermal characteristics of the building and the environment. A detailed model of such systems can be found in \cite{li2011optimal}.

\subsection{Network Model} \label{sect:network}

For a phase $\phi$ at bus $k$ which is assumed to be a member of $\mathcal{R}$, $\mathcal{B}$, and $\mathcal{H}$, the total power injection at time $t$ is given by:
\begin{align}\label{eq:nodal-inj}
    p_{k, \phi}(t) &= p_{k, \phi}^{(B)}(t) - p_{k, \phi}^{(R)}(t) - p_{k, \phi}^{(h)}(t) - p_{k, \phi}^{(L)}(t) \\
    q_{k, \phi}(t) &= q_{k, \phi}^{(B)}(t) - q_{k, \phi}^{(R)}(t) - q_{k, \phi}^{(h)}(t) - q_{k, \phi}^{(L)}(t)
\end{align}
We can assume that respective device power injection $p_{k,\phi}^{(\cdot)}$, $q_{k,\phi}^{(\cdot)}$ is zero when the corresponding DER is not present at the phase. Then, the active (reactive) power injections from the delta and wye phases at all buses $k \in \mathcal{N}^+$ at all time slots $t \in \{1, \ldots, T\}$ are collected in the vectors ${\vect{p}}_{\Delta}$ (${\vect{q}}_{\Delta}$) and ${\vect{p}}_{Y}$ (${\vect{q}}_{Y}$), respectively. 

The linear power flow model developed in~\cite{bernstein2017linear} is utilized to approximate the network-wide voltage magnitude profile as follows:
\begin{align}\label{eq:power-flow-voltMag}
    {\bf v} = {\bf M}_{\Delta}^{(p)} {\vect{p}}_{\Delta} + {\bf M}_{\Delta}^{(q)} {\vect{q}}_{\Delta} + {\bf M}_{Y}^{(p)} {\vect{p}}_{Y} + {\bf M}_{Y}^{(q)} {\vect{q}}_{Y} + \tilde{\bf v}
\end{align}
where ${\bf v}\in \mathbb{R}^{3NT}$ collects the voltage magnitudes at all phases $\phi \in \Phi$ at all buses $k \in \mathcal{N}^+$ for all time steps $t \in \{1, \ldots, T\}$. Then, constraints on the voltage magnitudes ${\bf v}$ are enforced as follows:
\begin{align} \label{eq:voltMag-const}
    \ubar{\bf v} \leq {\bf v} \leq \bar{\bf v}.
\end{align}
Furthermore, the net power injections at the three phases of the substation bus (Node $0$) are given by:
\begin{align}
       {\vect{p}}_0 = {\bf G}_{\Delta}^{(p)} {\vect{p}}_{\Delta} + {\bf G}_{\Delta}^{(q)} {\vect{q}}_{\Delta} + {\bf G}_{Y}^{(p)} {\vect{p}}_{Y} + {\bf G}_{Y}^{(q)} {\vect{q}}_{Y} + {\bf c} 
\end{align}
where the vector ${\vect{p}}_0$ collects the net injections at all phases at the substation for all time instants. 

The operational profiles of the considered DERs, which include $p_{k, \phi}^{(B)} (t)$, $q_{k, \phi}^{(B)} (t)$, $p_{k, \phi}^{(R)} (t)$, $q_{k, \phi}^{(R)} (t)$, $p_{k, \phi}^{(h)} (t)$, and $q_{k, \phi}^{(h)} (t)$, are collected in a vector $\vect{p} \in \mathbb{R}^{nT}$. After some simple manipulations, the constraints can be written compactly as:
\begin{align} \label{eq:LPfeasibility}
    \vect{W} {\vect{p}} \le \vect{u},
\end{align}
where $\vect{W}$ is a constant constraint matrix capturing the device operational constraints and the network voltage magnitude constraints for all time steps, and $\vect{u}$ incorporates the passive loads at each time period which are assumed to be known.

Moreover, the DERs operations and the power injection at the substation are coupled by the following equation:
\begin{align} \label{eq:aggragation}
    {\vect{p}}_0 = \vect{D}{\vect{p}} + \vect{b}
\end{align}
where $\vect{D} \in \mathbb{R}^{T \times nT}$ and $\vect{b} \in \mathbb{R}^T$ model the relationship between the substation power profile and the individual DER power profiles, i.e., \emph{the aggregation model}. Also, it is worth noting that the matrix $\vect{D}$ and vector $\vect{b}$ are both constant.

\section{Problem Statement and Optimization Formulation} \label{sect:formulation}




We assume that, irrespective of the generation type, location (bus number), and time of the day, the aggregator offers a flat buying and selling prices of $c_b$ and $c_s$, respectively, for unit power, where $c_s > c_b > 0$. For example, when the vector of bus power injection is $\vect{p}_t = \vect{p}^g_t - \vect{p}^\ell_t$ where $\vect{p}^g_t, \vect{p}^\ell_t \ge \mathbbold{0}$ are orthogonal vectors of bus power generation and demand at time $t$ (the orthogonality condition dictates that at most one of the two vectors is nonzero for each element), the payment to the customer by the system aggregator is $c_b(\vect{p}_t^g)^\top \mathbbold{1} - c_s(\vect{p}^{\ell}_t)^\top \mathbbold{1}$. Given the relationship between the aggregate power output $p^0_t$ and the individual bus power injection $\vect{p}_t$: $p_t^0 = \vect{D}_t\vect{p}_t + b_t$ and tfhe set of feasible $\vect{p}^0$ in $\mathcal{F} = \{ \vect{p}^0: \vect{p}^0 = \vect{D} \vect{p} + \vect{b}, \vect{W}\vect{p} \le \vect{u} \}$, we want to design a multi-time-step cost curve for $\vect{p}^0$ such that for any $\vect{p}^0 \in \mathcal{F}$, the total cost prescribed by the multi-time-step generation cost curves covers the maximum payment to customers by the system aggregator $c_b(\vect{p}_t^g)^\top \mathbbold{1} - c_s(\vect{p}^{\ell}_t)^\top \mathbbold{1}$ for \emph{any} feasible disaggregation strategies $(\vect{p}^g, \vect{p}^\ell)$.

\subsection{Linear Generation Cost Curve}

In this section we assume a linear generation cost curve with the form $\vect{y}^\top \vect{p}^0 + z$. With this model, the cost curve parameter optimization problem can be formulated as follows:
\begin{subequations} \label{eq:prob}
\begin{align}
    \min_{\vect{y}, z} \quad & f(\vect{y}, z) \\
    \text{s.t.} \quad & (\vect{p}^0)^\top \vect{y} + z \ge c_b (\vect{p}^g)^\top\mathbbold{1} - c_s (\vect{p}^\ell)^\top\mathbbold{1} \nonumber\\
    & \quad \forall (\vect{p}^0, \vect{p}^g, \vect{p}^\ell) \in \mathcal{D}^+, \label{eq:prob:b}
\end{align}
\end{subequations}
where 
\begin{multline}
    \mathcal{D}^+ = \big\{ (\vect{p}^0, \vect{p}^g, \vect{p}^\ell):\; \vect{p}^0 = \vect{D} \left( \vect{p}^g - \vect{p}^\ell \right) + \vect{b}, \\ 
    \vect{W}\left( \vect{p}^g - \vect{p}^\ell \right) \le \vect{u}, 
    \vect{p}^g, \vect{p}^\ell \ge \mathbbold{0}, \vect{p}^g \perp \vect{p}^\ell \big\}.
\end{multline}

The objective of the optimization problem is to minimize the volume under the linear cost curve over the box constraints of $p^0_t$ for all $t = 1, \ldots, T$, where the bounds on $p^0_t$ can be obtained by solving the optimization problems $\min/\max_{(\vect{p}^0, \vect{p}^g, \vect{p}^\ell) \in \mathcal{D}^+} p^0_t$. With the bounds on $\vect{p}^0$, the volume is given by the following integral:
\begin{align} \label{eq:linobj}
    & \phantom{=} \int\displaylimits_{\ubar{p}^0_T}^{\bar{p}^0_T} \cdots \int\displaylimits_{\ubar{p}^0_2}^{\bar{p}^0_2} \int\displaylimits_{\ubar{p}^0_1}^{\bar{p}^0_1} \left( \vect{y}^\top \vect{p}^0 + z \right) \mathrm{d}p^0_1\mathrm{d}p^0_2 \ldots \mathrm{d}p^0_T \nonumber\\
    & = \frac{1}{2} \sum_{i \in \mathcal{T}} \Bigg( y_i \left( (\bar{p}_i^0)^2 - (\ubar{p}_i^0)^2 \right) \prod_{\substack{k \in \mathcal{T} \\ k \neq i}} \left( \bar{p}_k^0 - \ubar{p}_k^0 \right) \Bigg) \nonumber\\
    & \hspace{2in} +z \prod_{k \in \mathcal{T}} \left( \bar{p}_k^0 - \ubar{p}_k^0 \right) \nonumber\\
    & = \frac{1}{2} \left( \prod_{k \in \mathcal{T}} \left( \bar{p}_k^0 - \ubar{p}_k^0 \right) \right) \left( \vect{y}^\top \left( \bar{\vect{p}}^0 + \ubar{\vect{p}}^0 \right) + 2z \right)
\end{align}
Since $\frac{1}{2} \prod_{k \in \mathcal{T}} \left( \bar{p}_k - \ubar{p}_k \right)$ is a positive constant, it suffices to minimize 
\begin{align}
    f(\vect{y}, z) := \vect{y}^\top \left( \bar{\vect{p}}^0 + \ubar{\vect{p}}^0 \right) + 2z
\end{align}
in place of \eqref{eq:linobj}.

It is seen that problem \eqref{eq:prob} is a robust linear programming problem in $(\vect{y},z)$ with uncertainty set $\mathcal{D}^+$. The complicating constraint in $\mathcal{D}^+$ is the orthogonality constraint $\vect{p}^g \perp \vect{p}^\ell$, without which the set is polyhedral. Problem \eqref{eq:prob} will therefore admit a computationally tractable reformulation when the constraint is absent \cite{ben2009robust}. Let $\mathcal{D}$ be the set without the orthogonality constraint:
\begin{multline}
    \mathcal{D} = \big\{ (\vect{p}^0, \vect{p}^g, \vect{p}^\ell):\; \vect{p}^0 = \vect{D} \left( \vect{p}^g - \vect{p}^\ell \right) + \vect{b}, \\ 
    \vect{W}\left( \vect{p}^g - \vect{p}^\ell \right) \le \vect{u}, 
    \vect{p}^g, \vect{p}^\ell \ge \mathbbold{0} \big\}.
\end{multline}
Since $\mathcal{D}^+ \subseteq \mathcal{D}$, replacing $\mathcal{D}^+$ by $\mathcal{D}$ in \eqref{eq:prob} results in a restriction of the original problem as:
\begin{subequations} \label{eq:restriction}
\begin{align}
    \min_{\vect{y}, z} \quad & f(\vect{y}, z) \\
    \text{s.t.} \quad & (\vect{p}^0)^\top \vect{y} + z \ge c_b (\vect{p}^g)^\top\mathbbold{1} - c_s (\vect{p}^\ell)^\top\mathbbold{1} \nonumber\\
    & \quad \forall (\vect{p}^0, \vect{p}^g, \vect{p}^\ell) \in \mathcal{D} \label{eq:restriction:b}
\end{align}
\end{subequations}

Fortunately, it is quite easy to see that \eqref{eq:restriction} is in fact an exact reformulation of \eqref{eq:prob}, which we state below:
\begin{pro}
    Any feasible solution to \eqref{eq:prob} is feasible to \eqref{eq:restriction}.  
\end{pro}
\begin{proof}
    Given $\vect{p}^g$ and $\vect{p}^\ell$ that are not perpendicular, there exists index $i$ such that $p_i^g, p_i^\ell > 0$. For every such $i$, reducing both $p_i^g$ and $p_i^\ell$ by $\min \{p_i^g, p_i^\ell\}$ so that one of them becomes zero. The new $\tilde{\vect{p}}^g$ and $\tilde{\vect{p}}^\ell$ are perpendicular. Since $c_s > c_b > 0$, $c_b (\tilde{\vect{p}}^g)^\top\mathbbold{1} - c_s (\tilde{\vect{p}}^\ell)^\top\mathbbold{1} \ge c_b (\vect{p}^g)^\top\mathbbold{1} - c_s (\vect{p}^\ell)^\top\mathbbold{1}$.
\end{proof}

The robust optimization problem can be rewritten as
\begin{subequations} \label{eq:twostage}
\begin{align}
    & \min_{\vect{y}, z} \; && f(\vect{y}, z) \\
    & \;\;\, \text{s.t.} \; && \min_{(\vect{p}^0, \vect{p}^g, \vect{p}^\ell) \in \mathcal{D}} (\vect{p}^0)^\top \vect{y} + z - c_b (\vect{p}^g)^\top\mathbbold{1} \nonumber\\
    & && \hspace{1.57in} + c_s (\vect{p}^\ell)^\top\mathbbold{1} \ge 0 \label{eq:twostage:b}
\end{align}
\end{subequations}
Suppose strong duality holds, the minimization problem in \eqref{eq:twostage:b} can be replaced by its dual problem, which is
\begin{subequations} \label{eq:dual}
\begin{align}
    \max_{\vect{\lambda}, \vect{\mu}} \quad & z - \vect{\lambda}^\top \vect{b} -\vect{\mu}^\top \vect{u} \\
    \text{s.t.} \quad & \vect{\lambda} - \vect{y} = 0 \\
    & \begin{bmatrix} \vect{D}^\top \\ -\vect{D}^\top \end{bmatrix} \vect{\lambda} + \begin{bmatrix} \vect{W}^\top \\ -\vect{W}^\top \end{bmatrix} \vect{\mu} \ge \begin{bmatrix} c_b\mathbbold{1} \\ -c_s\mathbbold{1} \end{bmatrix}
\end{align}
\end{subequations}
where $\vect{\lambda} \in \mathbb{R}^T$ and $\vect{\mu} \in \mathbb{R}^{mT}$ are dual variables corresponding to the constraints in $\mathcal{D}$. By replacing the minimization problem in \eqref{eq:twostage:b} with \eqref{eq:dual} and dropping the $\max$ sign, the cost curve parameter optimization problem \eqref{eq:prob} can be reformulated as the following linear program:
\begin{subequations} \label{eq:linear}
\begin{align}
    \min_{\vect{y}, z, \vect{\lambda}, \vect{\mu}} \quad & f(\vect{y}, z) \\
    \text{s.t.} \quad & z - \vect{\lambda}^\top \vect{b} -\vect{\mu}^\top \vect{u} \ge 0 \\
    & \vect{\lambda} - \vect{y} = 0 \\
    & \begin{bmatrix} \vect{D}^\top \\ -\vect{D}^\top \end{bmatrix} \vect{\lambda} + \begin{bmatrix} \vect{W}^\top \\ -\vect{W}^\top \end{bmatrix} \vect{\mu} \ge \begin{bmatrix} c_b\mathbbold{1} \\ -c_s\mathbbold{1} \end{bmatrix}.
\end{align}
\end{subequations}

\subsection{Quadratic Generation Cost Curve}


To improve the approximation quality of the linear cost curve, a quadratic one can be developed. Empirical observation on the maximum customer payment curves suggests the payment curves are likely concave. This makes intuitive sense: the net power consumption of DERs increases as the substation power $p^0_t$ increases, which leads to a more negative marginal cost (DERs' marginal payment increases) since the power selling price is higher than the buying price. We therefore restrict the quadratic cost curve to be concave. As we will see later in the section, this restriction also facilitates efficient solution algorithm.

We minimize the volume under the quadratic cost curve over the box constraints of $p^0_t$ for all $t \in \mathcal{T}$, which is given by the following integral:
\begin{align} \label{eq:quadobj}
    & \phantom{=} \int\displaylimits_{\ubar{p}^0_T}^{\bar{p}^0_T} \cdots \int\displaylimits_{\ubar{p}^0_2}^{\bar{p}^0_2} \int\displaylimits_{\ubar{p}^0_1}^{\bar{p}^0_1} \left( (\vect{p}^0)^\top \vect{Q} \vect{p}^0 + \vect{y}^\top \vect{p}^0 + z \right) \mathrm{d}p^0_1\mathrm{d}p^0_2 \ldots \mathrm{d}p^0_T \nonumber\\
    & = \frac{1}{3} \sum_{i \in \mathcal{T}} \Bigg( Q_{ii}\left( (\bar{p}_i^0)^3 - (\ubar{p}_i^0)^3 \right) \prod_{\substack{k \in \mathcal{T} \\ k \neq i}} \left( \bar{p}_k^0 - \ubar{p}_k^0 \right) \Bigg) \nonumber\\
    & \phantom{=} + \frac{1}{2} \sum_{\substack{i,j \in \mathcal{T} \\ i < j}} \Bigg( Q_{ij} \left( (\bar{p}_i^0)^2 - (\ubar{p}_i^0)^2 \right) \left( (\bar{p}_j^0)^2 - (\ubar{p}_j^0)^2 \right) \nonumber\\
    & \phantom{=} \prod_{\substack{k \in \mathcal{T} \\ k \neq i, j}} \left( \bar{p}_k^0 - \ubar{p}_k^0 \right) \Bigg) + z \prod_{k \in \mathcal{T}} \left( \bar{p}_k^0 - \ubar{p}_k^0 \right) \nonumber\\
    & \phantom{=} + \frac{1}{2} \sum_{i \in \mathcal{T}} \Bigg( y_i \left( (\bar{p}_i^0)^2 - (\ubar{p}_i^0)^2 \right) \prod_{\substack{k \in \mathcal{T} \\ k \neq i}} \left( \bar{p}_k^0 - \ubar{p}_k^0 \right) \Bigg) \nonumber\\
    & = \frac{1}{6} \left( \prod_{k \in \mathcal{T}} \left( \bar{p}_k^0 - \ubar{p}_k^0 \right) \right) \Bigg( 2\sum_{i \in \mathcal{T}} Q_{ii} \left( (\bar{p}_i^0)^2 + \bar{p}_i^0\ubar{p}_i^0 + (\ubar{p}_i^0)^2 \right) \nonumber\\
    & + 3\sum_{\substack{i,j \in \mathcal{T} \\ i < j}} Q_{ij} \left( \bar{p}_i^0 + \ubar{p}_i^0 \right) \left( \bar{p}_j^0 + \ubar{p}_j^0 \right) + 6z + 3 \vect{y}^\top \left( \bar{\vect{p}}^0 + \ubar{\vect{p}}^0 \right) \Bigg)
\end{align}
Since $\frac{1}{6} \prod_{k \in \mathcal{T}} \left( \bar{p}_k - \ubar{p}_k \right)$ is a positive constant, it suffices to minimize 
\begin{multline}
    g(\vect{Q}, \vect{y}, z) := 3\sum_{\substack{i,j \in \mathcal{T} \\ i < j}} Q_{ij} \left( \bar{p}_i^0 + \ubar{p}_i^0 \right) \left( \bar{p}_j^0 + \ubar{p}_j^0 \right) \\
    + 2\sum_{i \in \mathcal{T}} Q_{ii} \left( (\bar{p}_i^0)^2 + \bar{p}_i^0\ubar{p}_i^0 + (\ubar{p}_i^0)^2 \right) + 3 \vect{y}^\top \left( \bar{\vect{p}}^0 + \ubar{\vect{p}}^0 \right) + 6z
\end{multline}
in place of \eqref{eq:quadobj}.

With the objective function $g(\vect{Q}, \vect{y}, z)$, the quadratic cost curve design problem can be formulated as the following robust optimization problem:
\begin{subequations} \label{eq:quadorig}
\begin{align} 
    \min_{\vect{Q} \preceq 0, \vect{y}, z} \quad & g(\vect{Q}, \vect{y}, z) \label{eq:quadorig:a} \\
    \text{s.t.} \quad & (\vect{p}^0)^\top \vect{Q} \vect{p}^0 + \vect{y}^\top \vect{p}^0 + z \ge c_b (\vect{p}^g)^\top\mathbbold{1} - c_s (\vect{p}^\ell)^\top\mathbbold{1} \nonumber\\
    & \quad \forall (\vect{p}^0, \vect{p}^g, \vect{p}^\ell) \in \mathcal{D} \label{eq:quadorig:b}
\end{align}
\end{subequations}
It is tempting to repeat the approach used in the linear case to reformulate the robust constraint \eqref{eq:quadorig:b} through Lagrangian duality. However, strong duality no longer holds since $\vect{Q}$ is negative semidefinite. To see this, note the robust constraint \eqref{eq:quadorig:b} can be reformulated as follows:
\begin{multline} \label{eq:quadmin}
    \min_{(\vect{p}^0, \vect{p}^g, \vect{p}^\ell) \in \mathcal{D}} \big\{ (\vect{p}^0)^\top \vect{Q} \vect{p}^0 + \vect{y}^\top \vect{p}^0 + z \\
    - c_b (\vect{p}^g)^\top\mathbbold{1} + c_s (\vect{p}^\ell)^\top\mathbbold{1} \big\} \ge 0
\end{multline}
Since $\vect{Q} \preceq 0$, the quadratic program (QP) in \eqref{eq:quadmin} is concave. However, it is well known that the minimizer of a concave QP is attained at an extreme point of its feasible set \cite{floudas1995handbook}. If we denote the set of extreme points of $\mathcal{D}$ by $\mathcal{I}$, \eqref{eq:quadmin} admits the following equivalent reformulation:
\begin{multline}
    (\vect{p}^0_*)^\top \vect{Q} \vect{p}^0_* + \vect{y}^\top \vect{p}^0_* + z - c_b (\vect{p}^g_*)^\top\mathbbold{1} + c_s (\vect{p}^\ell_*)^\top\mathbbold{1} \ge 0, \\
    \forall (\vect{p}^0_*, \vect{p}^g_*, \vect{p}^\ell_*) \in \mathcal{I}.
\end{multline}
Since $\mathcal{I}$ is a finite set, the semi-infinite problem \eqref{eq:quadorig} admits a reformulation with a finite number of constraints, albeit potentially exponential in the dimension of $\mathcal{D}$:
\begin{subequations} \label{eq:quadreform}
\begin{align} 
    \min_{\vect{Q} \preceq 0, \vect{y}, z} \quad & g(\vect{Q}, \vect{y}, z) \\
    \text{s.t.} \quad & (\vect{p}^0_*)^\top \vect{Q} \vect{p}^0_* + \vect{y}^\top \vect{p}^0_* + z - c_b (\vect{p}^g_*)^\top\mathbbold{1} \nonumber\\
    & \quad+ c_s (\vect{p}^\ell_*)^\top\mathbbold{1} \ge 0, \quad \forall (\vect{p}^0_*, \vect{p}^g_*, \vect{p}^\ell_*) \in \mathcal{I} \label{eq:quadreform:b}
\end{align}
\end{subequations}
The problem structure of \eqref{eq:quadreform} suggests a natural solution approach based on constraint generation. We describe the algorithm below:

Initialization: Set the number of iteration $k=1$. Set convergence tolerance $\epsilon < 0$. Let $(\vect{p}_l^0, \vect{p}^l_l, \vect{p}_l^\ell), l \le L$ be the extreme points of $\mathcal{D}$ and let $\mathcal{I}_1 = \{ (\vect{p}_l^0, \vect{p}^g_l, \vect{p}_l^\ell), l \le L \}$.

Iteration $k > 1$:

Step 1) Solve the master problem.
The master problem is the following semidefinite program:
\begin{subequations} \label{eq:master}
\begin{align} 
    \min_{\vect{Q}, \vect{y}, z} \quad & g(\vect{Q}, \vect{y}, z) \label{eq:master:a} \\
    \text{s.t.} \quad & \Tr\left( \left(\vect{p}^0_*(\vect{p}^0_*)^\top\right) \cdot \vect{Q}\right) + \left( \vect{p}^0_* \right)^\top \vect{y} + z - c_b (\vect{p}^g_*)^\top\mathbbold{1} \nonumber\\
    & \quad+ c_s (\vect{p}^\ell_*)^\top\mathbbold{1} \ge 0, \quad \forall (\vect{p}^0_*, \vect{p}^g_*, \vect{p}^\ell_*) \in \mathcal{I}_k \label{eq:master:b} \\
    & \vect{Q} + \tilde{\vect{Q}} = 0, \label{eq:master:c} \\
    & \tilde{\vect{Q}} \succeq 0. \label{eq:master:d}
\end{align}
\end{subequations}
Let $(\vect{Q}_k, \vect{y}_k, z_k)$ be the optimal solution. 

Step 2) Solve the subproblem. The subproblem is the following concave QP:
\begin{subequations} \label{eq:sub}
\begin{align} 
    \min_{\vect{p}^0, \vect{p}^g, \vect{p}^{\ell}} \quad & (\vect{p}^0)^\top \vect{Q}_k \vect{p}^0 + \vect{y}_k^\top \vect{p}^0 + z_k \nonumber\\
    & \hspace{0.8in} - c_b (\vect{p}^g)^\top\mathbbold{1} + c_s (\vect{p}^\ell)^\top\mathbbold{1} \label{eq:sub:a} \\
    \text{s.t.} \quad & \vect{p}^0 = \vect{D}(\vect{p}^g - \vect{p}^\ell) + \vect{b}, \label{eq:sub:b} \\
    & \vect{W}(\vect{p}^g - \vect{p}^\ell) \le \vect{u}, \label{eq:sub:c} \\
    & \vect{p}^g, \vect{p}^\ell \ge 0. \label{eq:sub:d}
\end{align}
\end{subequations}
Several algorithms exist to solve concave QP to global optimality. Commercial solvers such as Gurobi and CPLEX are also capable of finding the global optimum. Assume such solver is at our disposal and let the globally optimal solution be $(\vect{p}_k^0, \vect{p}^g_k, \vect{p}_k^\ell)$. Let the optimal cost be $S_k$. 

Step 3) Check the convergence. If $S_k \ge \epsilon$, stop and return $(\vect{Q}_k, \vect{y}_k, z)$. Otherwise, set $\mathcal{I}_{k+1} = \mathcal{I}_k \cup \{ (\vect{p}_k^0, \vect{p}^g_k, \vect{p}_k^\ell) \}$, let $k = k+1$, and return to step 1.

\subsection{Derivation of Hourly Generation Cost Curve }

The previous two subsections discuss the design of multi-time-step generation cost curves. Given the substation power $\vect{p}^0$ across time steps $t = 1, \ldots, T$, these curves provide upper bounds on the payments to customers that realize the specific $\vect{p}^0$ considering the DER and network feasibility constraints over $t = 1, \ldots, T$. The multi-time-step cost curves are convenient when considering the operations of VPPs across multiple time steps. However, it may be more convenient to use the hourly generation cost curves for certain scenarios when short-term dispatch commands are sought after. Solving the optimization problems in the previous subsections for $T = 1$ yield hourly cost curves. However, these curves may be overly conservative as the optimization problem may require constraint satisfaction for some infeasible $(p^0_t, \vect{p}^g_t, \vect{p}^\ell_t)$ instances, since no inter-temporal constraints which further constrains the DER output are considered. To incorporate the inter-temporal constraints in the design of the hourly cost curve, the optimization problems need slight modification. Take the quadratic case as an example (the linear case is identical). Instead of solving \eqref{eq:quadorig} for $T = 1$, we solve:
\begin{subequations} \label{eq:quadsingle}
\begin{align} 
    \min_{Q < 0, y, z} \quad & g(Q, y, z) \label{eq:quadsingle:a} \\
    \text{s.t.} \quad & Q(p^0_t)^2 + yp^0_t + z \ge c_b (\vect{p}^g_t)^\top\mathbbold{1} - c_s (\vect{p}^\ell_t)^\top\mathbbold{1}, \nonumber\\
    & \quad \forall \left( \begin{bsmallmatrix}p^0_1 \\ \vdots \\ p^0_t \\ \vdots \\ p^0_T \end{bsmallmatrix}, \begin{bsmallmatrix} \vect{p}^g_1 \\ \vdots \\ \vect{p}^g_t \\ \vdots \\ \vect{p}^g_T \end{bsmallmatrix}, \begin{bsmallmatrix} \vect{p}^\ell_1 \\ \vdots \\ \vect{p}^\ell_t \\ \vdots \\ \vect{p}^\ell_T \end{bsmallmatrix} \right) \in \mathcal{D}, \label{eq:quadsingle:b}
\end{align}
\end{subequations}
That is, we keep the objective and the inequality \eqref{eq:quadorig:b} intact, while requiring that the inequality \eqref{eq:quadsingle:b} only holds for those $(p^0_t, \vect{p}_t^g, \vect{p}_t^\ell)$ that are parts of feasible time series $(\vect{p}^0, \vect{p}^g, \vect{p}^\ell)$ that spans over some time horizon $t = 1, \ldots, T$, so that the inter-temporal DER constraints are incorporated.

The other way to derive the hourly cost curve considering inter-temporal constraints given the multi-time-step cost curve is to project the cost curve onto a specific hour. However, finding this projection is in general computationally hard, and the projection does not necessarily preserve the linear or quadratic form of the original cost curve.

\section{Numerical Simulation Results} \label{sect:simulation}

The proposed formulations are tested using data of a real
distribution feeder located in the territory of Southern California Edison. The distribution feeder has 126 multiphase nodes with a total of 366 single-phase points of connection. The nominal voltage at the substation is 12 kV, and voltage limits are set to 1.05 p.u. and 0.95 p.u for all nodes. There are 55 uncontrollable loads scattered across the feeder. Dispatchable DERs include 33 PV units, 28 energy storage devices, and 5 HVAC systems. Detailed configurations and parameters of the distribution feeder can be found in \cite{bernstein2019real}. The device model parameters follow \cite{chen2019aggregate}, except that we scale the capacities of the DERs to ensure the feeder is not only a power consumer but is also capable of operating as a VPP that provides power support to the upstream system. Specifically, the PV units are scaled up 10 times, and the capacities of the energy storage units and the HVAC systems are scaled up 5 times. We assume for the aggregator, the per unit power buying price for each DER is $c_b = 1$ and the per unit power selling price for each DER is $c_s = 2$. For simplicity the price of the constant load powers are not taken into account. MATLAB 2020b is used for all computational experiments. CVX \cite{cvx} are used to formulate convex optimization problems and all optimization problems are solved with Gurobi 9.1\cite{gurobi}.

\subsection{Multi-Time-Step Linear and Quadratic Cost Curves}

In the case study, we implement both the linear and quadratic cost curve generation formulations described in Section \ref{sect:formulation}. We set the time horizon to be from 9:00 to 13:00 with 1-hour granularity, so the number of time steps is four. The quadratic cost curve is expected to dominate the linear one in terms of the optimal objective value since for an optimal solution $(\vect{y}^*, z^*)$ of the linear case, the solution $(\mathbbold{0}, \vect{y}^*, z^*)$ is always a feasible solution for the quadratic case.

We obtain the quadratic and linear cost curves by implementing the respective solution algorithms detailed in Section \ref{sect:formulation}. The two cost curves are given by
\begin{subequations}
\begin{align}
    & \textbf{Quadratic:} \nonumber\\
    & q(\vect{p}^0) = (\vect{p}^0)^\top 10^{-6}
    \begin{bmatrix}
    0.52 & 0.05 & 0.06 & 0.06 \\
    0.05 & 0.07 & 0.01 & -0.00 \\
    0.06 & 0.01 & 0.04 & 0.01 \\
    0.06 & -0.00 & 0.01 & 0.04
    \end{bmatrix} \vect{p}^0 \nonumber\\
    & - 10^3\begin{bmatrix}
    6.85 & 6.85 & 6.85 & 6.85
    \end{bmatrix} \vect{p}^0 + 1.62\times10^5, \label{eq:quadcost} \\
    & \textbf{Linear:} \nonumber\\
    & l(\vect{p}^0) = - 10^3\begin{bmatrix}
    6.85 & 6.85 & 6.85 & 6.85
    \end{bmatrix} \vect{p}^0 + 1.62\times10^5. \label{eq:lincost}
\end{align}
\end{subequations}
The optimal costs $f^*$ and $g^*$ for problems \eqref{eq:linear} and \eqref{eq:quadorig} are exactly the same after scaling, with $f^* = g^*/3= 5.89 \times 10^5$.

By looking at the coefficients of both cost curves, we see that both the linear coefficients and the constant terms are exactly the same between the two curves. In addition, the elements of the quadratic matrix in \eqref{eq:quadcost} are negligible. Under the premise of a moderate loading conditions and a normal operation condition under which the linearized power flow approximation is relatively accurate, the similarity between the two curves is a reflection of the relative `flatness' of the maximum payment curve. However, we do expect this to change as the nonlinearity of the payment curve becomes more pronounced due to either more stressed loading conditions, more stringent DER/operating constraints, or shorter time periods.

\subsection{Convergence of the Constraint Generation Algorithm}

As discussed in Section \ref{sect:formulation}, the quadratic cost curve design problem is solved in a constraint generation algorithm. For the numerical problem considered above, the algorithm converges in 35 iterations, where the convergence tolerance is set to be $\epsilon = 1 \times 10^{-6}$. Fig. \ref{fig:convergence} shows the optimal cost of the subproblem until convergence. The algorithm converges quite fast on the problem with four time steps, with each iteration taking less than 1 second. 

\begin{figure}[!t]
    \centering
    \includegraphics[width=0.9\linewidth]{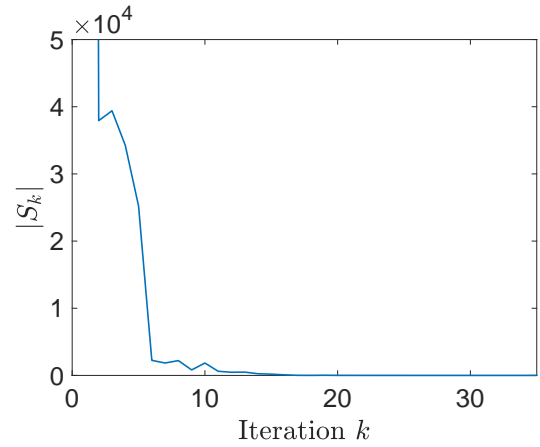}
	\caption{Convergence of the constraint generation algorithm for quadratic cost curve optimization.} 
	\label{fig:convergence}
\end{figure}

\subsection{Hourly Linear and Quadratic Generation Cost Curve} \label{sect:hourly}

We next examine the design of hourly cost curves with the time horizon of 4 hours as considered before. Fig. \ref{fig:hourly} shows the four linear and quadratic hourly cost curves. The curve of the maximum payment to customers is calculated numerically by maximizing $c_b (\vect{p}^g_t)^\top\mathbbold{1} - c_s (\vect{p}^\ell_t)^\top\mathbbold{1}$ for a fixed $p^0_t$. The numerical maximum payment curve do suggest that it is concave and nonlinear, for which the quadratic cost curve approximations achieve more accurate approximation result at the cost of introducing some nonconvexity. It is worth noting that since the maximum payment curves are very close to being piecewise linear for all four hours, it would be an interesting research direction to consider deriving piecewise linear cost curves, potentially through robust mixed-integer programming formulations.

\begin{figure}[!t]
    \centering
    \captionsetup[subfigure]{justification=centering}
     \begin{subfigure}{0.24\textwidth}
        \includegraphics[width=1.05\linewidth]{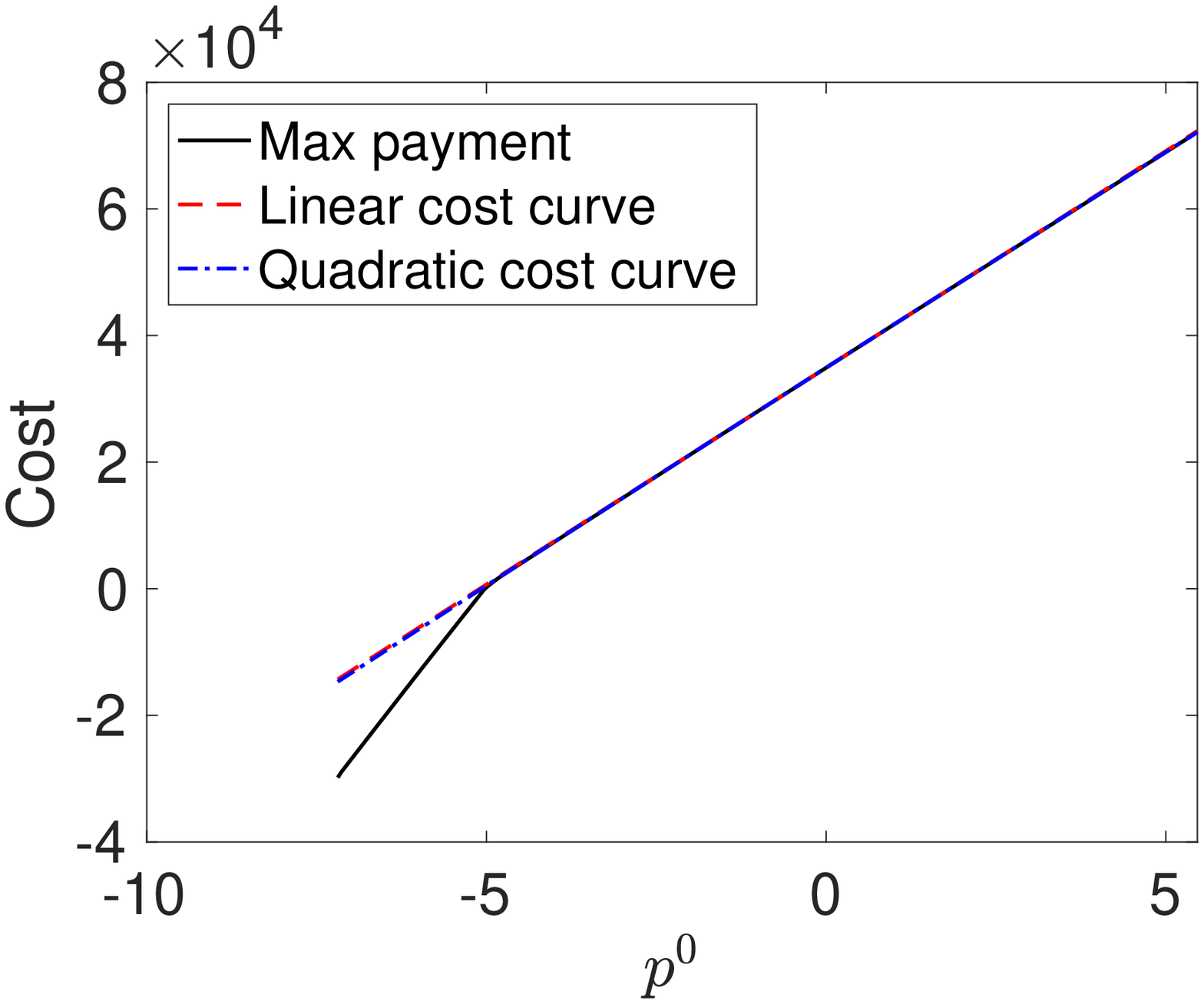}
	    \caption{9:00--10:00.} 
	    \label{fig:proj:1}
     \end{subfigure}
     \begin{subfigure}{0.24\textwidth}
	    \includegraphics[width=1.05\linewidth]{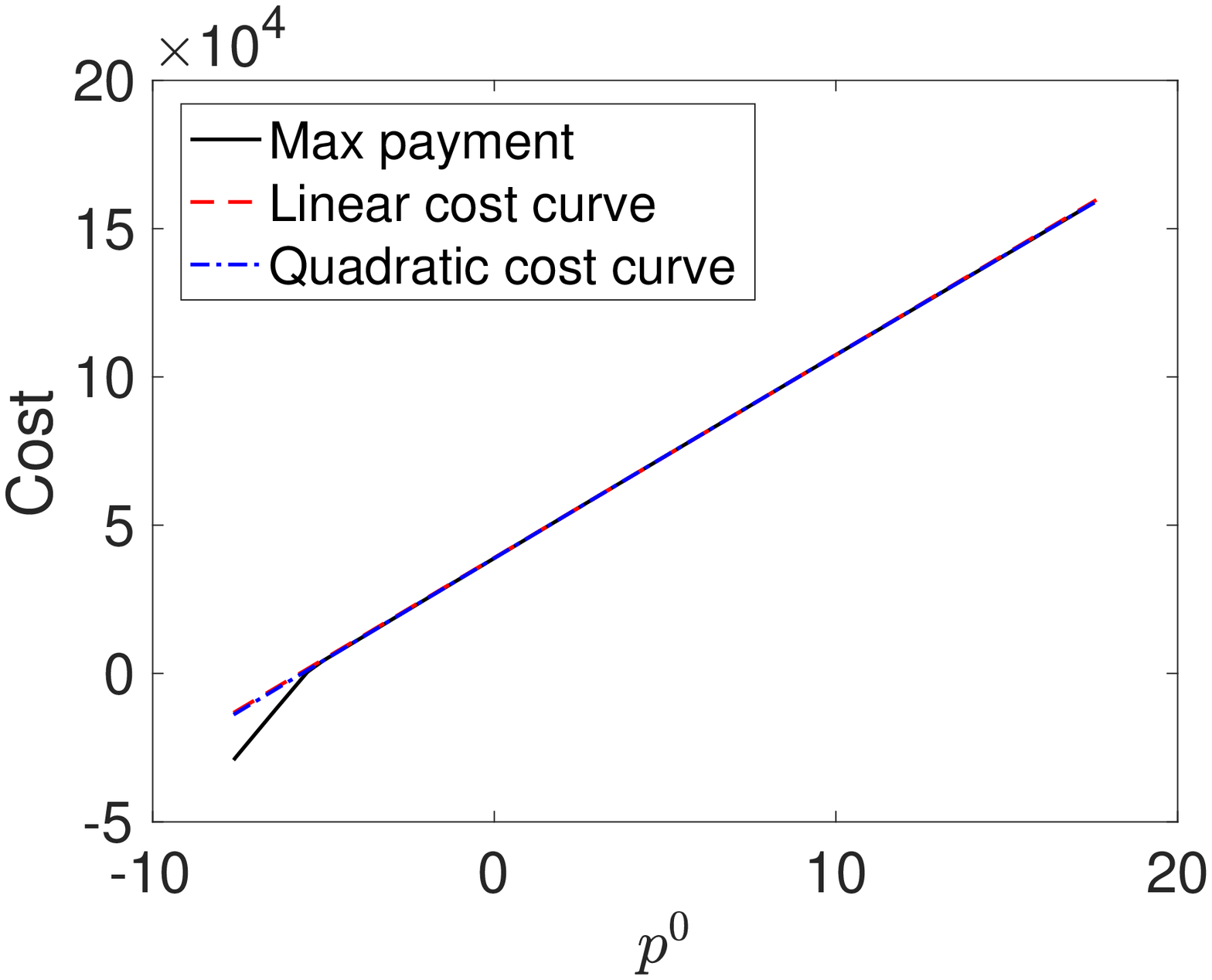}	
	    \caption{10:00--11:00.} 
	    \label{fig:proj:2}
     \end{subfigure}
     \begin{subfigure}{0.24\textwidth}
        \centering
	    \includegraphics[width=1.05\linewidth]{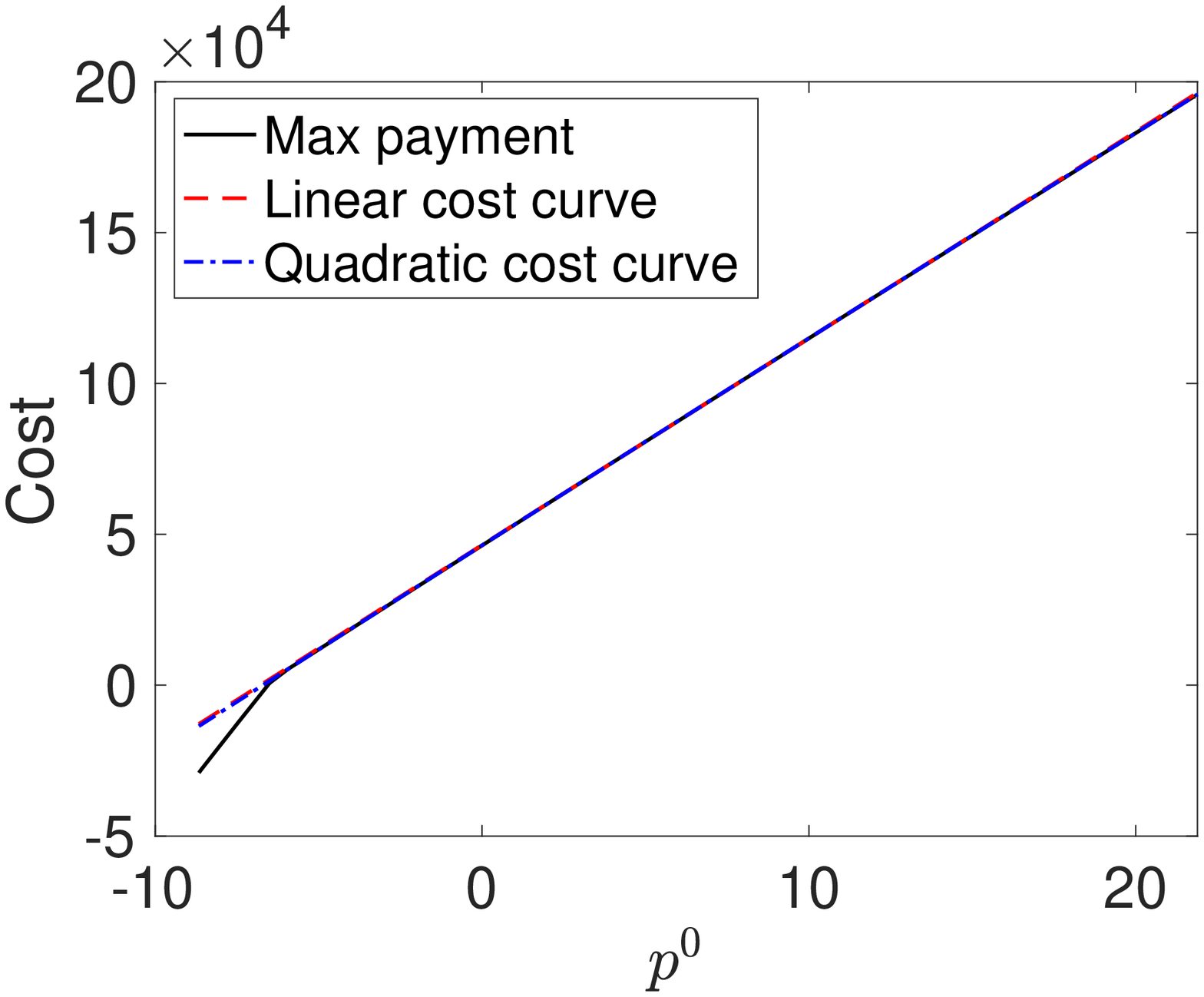}	
	    \caption{11:00--12:00.} 
	    \label{fig:proj:3}
     \end{subfigure}
     \begin{subfigure}{0.24\textwidth}
        \centering
	    \includegraphics[width=1.05\linewidth]{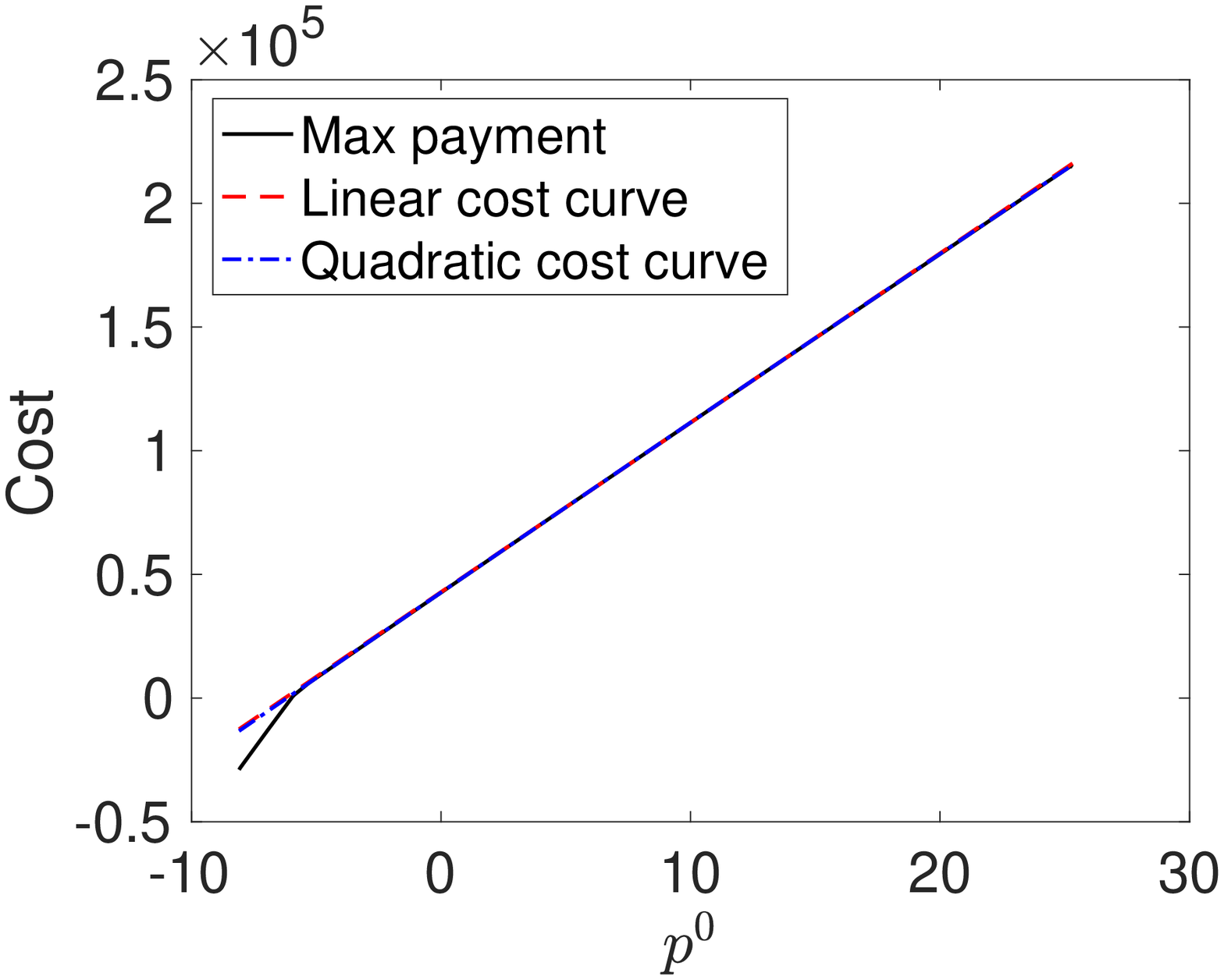}	
	    \caption{12:00--13:00.} 
	    \label{fig:proj:4}
     \end{subfigure}
     \caption{Comparison of linear and quadratic cost curves as upper bounds of the curve of the maximum payment to customers.}\vspace{-12pt}
     \label{fig:hourly}
\end{figure}

\section{Conclusions} \label{sect:conclusion}
This paper proposed efficient optimization formulations and solution approaches for the characterization of multi-time-step as well as hourly generation cost curves for a distribution system with high penetration of DERs. Network and DER constraints are taken into account when deriving these cost curves, and they enable active distribution systems to bid into the electricity market. The problems of deriving linear and quadratic cost curves are formulated as robust optimization problems and tractable reformulation/solution algorithm are developed to facilitate efficient calculations. Simulation results validated the effectiveness of the proposed formulations and the solution algorithm, and demonstrated the quality of the derived generation cost curves as tight upper bounds to the maximum payment to the customers. Some future research directions include alternative cost curves (piecewise linear), nonlinear network models, and the incorporation of network and parameter uncertainties.



%



\bibliographystyle{IEEEtran}
\bibliography{irep2022_market.bib}

\end{document}